\documentclass[conference]{IEEEtran}
\usepackage[colorlinks,
linkcolor=red,
anchorcolor=green,
citecolor=blue
]{hyperref} 
\usepackage{amsfonts,amssymb,array,amsmath}
\usepackage{latexsym}
\usepackage{CJK}
\usepackage{cite}
\usepackage{bm}
\usepackage{soul}

\usepackage{graphicx}
\usepackage{epstopdf}
\usepackage{epsfig}
\usepackage{subfigure} 

\usepackage{verbatim}  

\usepackage[usenames,dvipsnames]{color}
\definecolor{mycyan}{cmyk}{.3,0,0,0}
\usepackage{colortbl}
\usepackage{enumerate}
\usepackage{cases}

\usepackage{mathrsfs}
\usepackage{booktabs}
\usepackage{textcomp}
\usepackage{multirow}
\usepackage{lettrine}
\usepackage[scaled=0.92]{helvet}
\usepackage[utf8]{inputenc} 
\usepackage[T1]{fontenc}
\usepackage{url}
\usepackage{ifthen}
\usepackage{algorithm}
\usepackage{algorithmic}

\usepackage{amsthm} 

\newtheorem{thm}{Theorem}
\newtheorem{defn}{Definition}
\newtheorem{rem}{Remark}

\begin{document}
%
\title{How Many Samples Required in Big Data Collection: A Differential Message Importance Measure}
%
%
%

\author{Shanyun Liu, ~Rui~She, ~Pingyi~Fan\\

\small
Tsinghua National Laboratory for Information Science and Technology(TNList),\\
Department of Electronic Engineering, Tsinghua University, Beijing, P.R. China\\
E-mail: liushany16@mails.tsinghua.edu.cn, sher15@mails.tsinghua.edu.cn, ~fpy@tsinghua.edu.cn}

\maketitle

\begin{abstract}
Information collection is a fundamental problem in big data, where the size of sampling sets plays a very important role. This work considers the information collection process by taking message importance into account. Similar to differential entropy, we define differential message importance measure (DMIM) as a measure of message importance for continuous random variable. It is proved that the change of DMIM can describe the gap between the distribution of a set of sample values and a theoretical distribution. In fact, the deviation of DMIM is equivalent to Kolmogorov-Smirnov statistic, but it offers a new way to characterize the distribution goodness-of-fit. Numerical results show some basic properties of DMIM and the accuracy of the proposed approximate values. Furthermore, it is also obtained that the empirical distribution approaches the real distribution with decreasing of the DMIM deviation, which contributes to the selection of suitable sampling points in actual system.
\end{abstract}


\IEEEpeerreviewmaketitle
\section{Introduction}\label{Sec:Introduction}
The actual system of big data needs to process lots of data within a limited time generally, so many researches are on sample data to improve their efficiency \cite{chen2014big}. Distribution goodness-of-fit plays a fundamental role in signal processing and information theory, which focuses on the error magnitude between the distribution of a set of sample values and the real distribution. This paper desires to solve this problem based on information theory.

Shannon entropy \cite{shannon2001mathematical} is possibly the most important quantity in information theory, which describes the fundamental laws of data compression and communication \cite{verdu1998fifty}. Due to its success, numerous entropies have been provided in order to extend information theory. Among them, the most successful expansion is R{\'e}nyi entropy \cite{renyi1961measures}. There are many applications based on R{\'e}nyi entropy, such as hypothesis testing \cite{van2014renyi}.

Actually, entropy is a quantity with respect to probability distribution, which satisfies the intuitive notion of what a measure of information should be \cite{Elements}. Therefore, in this paper, we propose differential message importance measure (DMIM) as a measure of information for continuous random variable to characterize the process of information collection. DMIM is expanded from discrete message importance measure (MIM) \cite{fan2016message} which is such an information quantity which agrees with the intuitive notion of information importance for small probability event. Recent studies show that MIM has many applications in big data, such as information divergence measures \cite{she2017amplifying} and compressed data storage \cite{liu2017non}.

Much of the research in the goodness of fit in the past several decades focused on the Kolmogorov-Smirnov test \cite{massey1951kolmogorov,lilliefors1967kolmogorov}. Based on it, \cite{Resnick1992Advantures} gave an
error estimation of empirical distribution. This result can describe the goodness of fit very well and guide us to choose the sampling numbers. However, it can not visually display the process of information collection because the previous results can not describe the message carried by each sample and the information changes with the increase of the sampling size. The problem of testing goodness-of-fit in a discrete setting was discussed in \cite{harremoes2012information}. Fortunately, DMIM is the proper measure to help us consider the problem of goodness of fit in the view of the information collection of continuous random variables. Moreover, Compared with Kolmogorov-Smirnov statistic, DMIM also shows the relationship between the variance of a random variable and the error estimation of empirical distribution.

The rest of this paper is organized as follows. Section \ref{sec:two} introduces the definition and basic properties of DMIM. Then, the DMIM of some basic continuous distribution is discussed in Section \ref{sec:three}, in which we give the asymptotic analysis of Gaussian distribution. In Section \ref{sec:four}, the goodness of fit with DMIM is discussed to analyze the process of information collection. The validity of proposed theoretical results is verified by the simulation results in Section \ref{sec:five}. Finally, we finish the paper with conclusions in Section \ref{sec:six}.

\section{The Definition and Properties of DMIM}\label{sec:two}
\subsection{Differential Message Important Measure}
In this part, a new measure of information for continuous random variable will be introduced, which is called DMIM. It is an extension of MIM.
 \begin{defn}\label{Non-parametric MIM}
The DMIM $l(X)$ of a continuous random variable $X$ with density $f(x)$ is defined as
\begin{equation}
l(X) = {\int_S {f(x){e^{ - f(x)}}dx} },
\end{equation}
where $S$ is the support set of the random variable.
\end{defn}
In fact, the DMIM of a continuous random variable $X$ with density $f(x)$ can be written as
\begin{flalign}\label{equ:Normal}
l(X)=&\int_{ - \infty }^{ + \infty } {f (x)} {e^{ - f (x)}}dx \\
=& \int_{ - \infty }^{ + \infty } {f (x)} \sum\limits_{ n= 0}^\infty  {{{{{\left( { - f (x)} \right)}^n}} \over {n!}}} dx\tag{\theequation a}\label{equ:Normal a}\\
 = &\int_{ - \infty }^{ + \infty } {\sum\limits_{n = 0}^\infty  {{{\left( { - 1} \right)}^n}{{{{\left( {f (x)} \right)}^{n + 1}}} \over {n!}}} dx}  \tag{\theequation b}\label{equ:Normal b} \\
  =& \int_{ - \infty }^{ + \infty } {f (x)dx + \sum\limits_{n = 1}^\infty  {\int_{ - \infty }^{ + \infty } {{{\left( { - 1} \right)}^n}{{{{\left( {f (x)} \right)}^{n + 1}}} \over {n!}}dx} } }\tag{\theequation c}\label{equ:Normal c}\\
   =& 1 + \sum\limits_{n = 1}^\infty  {{{{{\left( { - 1} \right)}^n}} \over {n!}}\int_{ - \infty }^{ + \infty } {{{\left( {f (x)} \right)}^{n + 1}}dx} }. \tag{\theequation d}\label{equ:Normal d}
\end{flalign}

\subsection{Properties of DMIM}
In this part, several basic properties of DMIM are discussed in details.
\subsubsection{Upper and Lower Bound}
For any continuous random variable $X$ with density $f(x)$, due to $0\leq f(x)\leq 1$, it is obtained that
\begin{flalign}\label{equ:Bound}
0\le \int_S {f(x){e^{ - f(x)}}dx}  \le \int_S {f(x)dx}  = 1.
\end{flalign}


\subsubsection{Translation}
Let $Y=X+c$, where $c$ is a constant. Then $f_Y(y)=f_X(y-c)$, and
\begin{equation}
l(X+c)=\int_{-\infty}^{+\infty} {{f_X}(x - c){e^{ - {f_X}(x - c)}}dx} =l(X).
\end{equation}
As a result, the constant drift does not change the DMIM.

\subsubsection{Relation of DMIM to R{\'{e}}nyi Entropy}
The differential R{\'{e}}nyi entropy of a continuous random variable $X$ with density $f(x)$ is given by \cite{van2014renyi}
\begin{equation}
 h_{\alpha}(X)=\frac{1}{1-\alpha} \ln \int {(f(x))}^{\alpha}dx,
\end{equation}
where $\alpha>0$ and $\alpha \ne 1$. As $\alpha$ tends to 1, the Rényi entropy tends to the Shannon entropy.

Therefore, we obtain
\begin{equation}\label{equ:H_renyi}
 \int {(f(x))}^{\alpha}dx=e^{(1-\alpha) h_{\alpha}(X)}.
\end{equation}
Hence, we find
\begin{equation}\label{equ:Renyi}
  l(X) = 1 + \sum\limits_{n = 1}^\infty  {{{{{\left( { - 1} \right)}^n}} \over {n!}}e^{-n h_{n+1}(X)} }.
\end{equation}
Obviously, the DMIM is an infinite series of R{\'{e}}nyi Entropy.

\subsubsection{Truncation Error}
In this part, the remainder term of $l(X)$ will be discussed.
\begin{thm}\label{thm:Truncation Error}
 If $\int {(f(x))}^{n+1}dx \leq \varepsilon $ for every $n\geq m$, then
 \begin{equation}\label{equ:Truncation Error}
\left| {l(X) - (1 + \sum\limits_{n = 1}^{m - 1} {{\frac{{(-1)}^n}{n!}}{e^{ - n{h_{n + 1}}(X)}}} )} \right| \le e\varepsilon.
 \end{equation}
\end{thm}
\begin{proof}
Substituting (\ref{equ:H_renyi}) and (\ref{equ:Renyi}) in the left of (\ref{equ:Truncation Error}), we obtain
\begin{flalign}\label{equ:Truncation Error 1}
&\left| {l\left( X \right) - \left( {1 + \sum\limits_{n = 1}^{m - 1} {{{{{\left( { - 1} \right)}^n}} \over {n!}}\int_{ - \infty }^{ + \infty } {{{\left( {f(x)} \right)}^{n + 1}}dx} } } \right)} \right| \nonumber\\
=&   \left| {\sum\limits_{n = m}^\infty  {{{{{\left( { - 1} \right)}^n}} \over {n!}}\int_{ - \infty }^{ + \infty } {{{\left( {f(x)} \right)}^{n + 1}}dx} } } \right| \\
\leq &  \sum\limits_{n = m}^\infty  {\left| {{{1} \over {n!}}\int_{ - \infty }^{ + \infty } {{{\left( {f(x)} \right)}^{n + 1}}dx} } \right|}  \tag{\theequation a}\label{equ:Truncation Error 1 }\\
\leq &\left( {\sum\limits_{n = m}^\infty  {{1 \over {n!}}} } \right)\varepsilon  \tag{\theequation b}\label{equ:Truncation Error 1 c}\\
\leq&\left( {1 + \sum\limits_{n = 1}^{m - 1} {{1 \over {n!}}}  + \sum\limits_{i = m}^\infty  {{1 \over {n!}}} } \right)\varepsilon =e\varepsilon  \tag{\theequation c}\label{equ:lTruncation Error 1 e}
\end{flalign}
where (\ref{equ:Truncation Error 1 c}) follows from $\int {(f(x))}^{n+1}dx \leq \varepsilon $.
\end{proof}
That is to say, if the integral of the density to the $(n+1)$-th power is limited, the remainder term will be restricted.

 \begin{rem}\label{rem:jieduan}
 Let $m=2$ in (\ref{equ:Truncation Error 1 c}), we obtain
 \begin{equation}
 \left| {l(x) - (1 -e^{-h_2(x)} )} \right| \le (1+1+\sum\limits_{n= 2}^{\infty}  {{  \frac{1}{n!}} } -2)  \varepsilon =(e-2)\varepsilon.
 \end{equation}
 \end{rem}


\section{The DMIM of Typical Distributions} \label{sec:three}
\subsection{Uniform Distribution}
For a random variable whose density is $\frac{1}{b-a}$ from $a$ to $b$ and $0$ elsewhere. Then we obtain
\begin{equation}\label{equ:Uniform 1}
l(X) =  {\int_a^b {{1 \over {b - a}}{e^{ - {1 \over {b - a}}}}dx} } = {e^{ - {1 \over {b - a}}}}.
\end{equation}
It is also noted that
\begin{flalign}\label{equ:Uniform 2}
\mathop {\lim }\limits_{\left( {b - a} \right) \to 0} {e^{ - {1 \over {b - a}}}} =0,\quad \mathop {\lim }\limits_{\left( {b - a} \right) \to \infty } {e^{ - {1 \over {b - a}}}} = 1.
\end{flalign}

\subsection{Normal Distribution}
Let $X\sim \phi (x) = {1 \over {\sqrt {2\pi {\sigma ^2}} }}{e^{ - {{{{\left( {x - \mu } \right)}^2}} \over {2{\sigma ^2}}}}}$, we obtain
\begin{flalign}\label{equ:Normal 1}
&\int_{ - \infty }^{ + \infty } {{{\left( {\phi (x)} \right)}^{n + 1}}dx}  = \int_{ - \infty }^{ + \infty } {{{\left( {{1 \over {\sqrt {2\pi {\sigma ^2}} }}{e^{ - {{{{\left( {x - \mu } \right)}^2}} \over {2{\sigma ^2}}}}}} \right)}^{n + 1}}dx} \\
 &={\left( {{1 \over {\sqrt {2\pi {\sigma ^2}} }}} \right)^{n + 1}}\int_{ - \infty }^{ + \infty } {{e^{ - {{n + 1} \over {2{\sigma ^2}}}{{\left( {x - \mu } \right)}^2}}}dx}  \tag{\theequation a}\label{equ:Normal 1 a}\\
 & = {\left( {{1 \over {\sqrt {2\pi {\sigma ^2}} }}} \right)^{n + 1}}\sqrt {{{2\pi {\sigma ^2}} \over {n + 1}}}    = {1 \over {\sqrt {n + 1} }}{\left( {{1 \over {\sqrt {2\pi {\sigma ^2}} }}} \right)^n}.  \tag{\theequation b}\label{equ:Normal 1 c}
   \end{flalign}
Substituting (\ref{equ:Normal 1 c}) in (\ref{equ:Normal d}), we obtian
\begin{equation}\label{equ:normal RMIM}
l(X) = 1 + \sum\limits_{n = 1}^\infty  {{{{{\left( { - 1} \right)}^n}} \over {n!}}} {1 \over {\sqrt {n + 1} }}{\left( {{1 \over {\sqrt {2\pi {\sigma ^2}} }}} \right)^n}.
\end{equation}
Obviously, if $\sigma>1/\sqrt{2 \pi}$, $\int_{ - \infty }^{ + \infty } {{{\left( {\phi (x)} \right)}^{n + 1}}dx}$ will be less than or equal to $1/{(2\sqrt{3}\pi\sigma^2})$ for every $n\geq 2$ because $ {1 \over {\sqrt {n + 1} }}{\left( {{1 \over {\sqrt {2\pi {\sigma ^2}} }}} \right)^n}$ monotonically decreases in this case. According to Remark \ref{rem:jieduan}, we obtain
\begin{equation}\label{equ:normal jieduan 1}
  \left| {l(x) - (1 -e^{-h_2(x)} )} \right| \leq \frac{(e-2) }{{2\sqrt{3}\pi\sigma^2}}\approx \frac{0.066}{\sigma^2}.
\end{equation}
If $\sigma$ is big enough, $\frac{(e-2) }{{2\sqrt{3}\pi\sigma^2}} \approx 0$. In this case, substituting $h_2(X)=\ln2+0.5\ln\pi +\ln \sigma$ in (\ref{equ:normal jieduan 1}), we find
\begin{equation}\label{equ:normal jieduan 2}
  l(x) \approx 1-\frac{1}{2\sqrt{\pi}\sigma } \approx e^{-\frac{1}{2\sqrt{\pi}\sigma}}.
\end{equation}
In fact, $e^{-\frac{1}{2\sqrt{\pi}\sigma}}$ is a very good approximation for DMIM of normal distribution when $\sigma$ is not too small, which will be shown by the numerical results in section \ref{sec:five}.


\subsection{Negative Exponential Distribution}
Letting
\begin{equation}
\begin{split}
 X\sim f(x)=\left\{
   \begin{aligned}
 & \lambda e^{-\lambda x}, \quad x \geq 0 \\
 & 0,\quad\quad\quad x<0 \\
   \end{aligned}
   \right.,
   \end{split}
\end{equation}
we obtain
\begin{flalign}\label{equ:Negative Exponential 1}
&\int_{ - \infty }^{ + \infty } {{{\left( {f(x)} \right)}^{n + 1}}dx}  = \int_0^{ + \infty } {{{\left( {\lambda {e^{ - \lambda x}}} \right)}^{n + 1}}dx}    = {{{\lambda ^n}} \over {n + 1}}.
 \end{flalign}
Substituting (\ref{equ:Negative Exponential 1}) in (\ref{equ:Normal d}), we obtain
\begin{equation}\label{equ: Negative Exponential RMIM}
l(X) = 1 + \sum\limits_{n = 1}^\infty  {{{\left( { - 1} \right)}^n}{{{\lambda ^n}} \over {\left( {n + 1} \right)!}}} =\frac{1}{\lambda}(1-e^{-\lambda}).
\end{equation}
It is noted that
\begin{flalign}\label{equ:Negative Exponential 2}
 \mathop {\lim }\limits_{\lambda  \to 0} {1 \over \lambda }\left( {1 - {e^{ - \lambda }}} \right) = 1, \quad \mathop {\lim }\limits_{\lambda  \to \infty} {1 \over \lambda }\left( {1 - {e^{ - \lambda }}} \right) = 0 .
 \end{flalign}

\section{Goodness of Fit} \label{sec:four}
In this section, we will consider the problem of distribution goodness-of-fit in a continuous setting. Let $X_1,X_2,...X_n$ be a sequence of independent and identically distributed random variables, each having mean $\mu$ and variance $\sigma^2$. In practice, the real distribution is generally unknown and we usually use empirical distribution to substitute real distribution. Generally, the empirical distribution function is given by
\begin{equation}
 \hat F_n(x)=\frac{1}{n}\sum\limits_{k=1}^n {I_{(X_k\leq x)}},
\end{equation}
and the real distribution is $F(x)$ .

One practical problem that can occur with this strategy is that how many samples is required for fitting the real distribution. Many literatures studied this problem by Kolmogorov-Smirnov statistic \cite{massey1951kolmogorov,lilliefors1967kolmogorov,Resnick1992Advantures}. When $n$ is big enough, the confidence limits for a cumulative distribution are given by \cite{Resnick1992Advantures},
\begin{equation}
 P\{D_n>d\}\approx 2\sum\limits_{k=1}^{\infty} {{(-1)}^{k-1}e^{-2nk^2d^2}},
\end{equation}
where $D_n$ is error bound between empirical distribution and real distribution, called Kolmogorov-Smirnov statistic, which is defined as
\begin{equation}
D_n=\mathop {\sup}\limits_x {\left| \hat F_n (x)-F(x)\right|},
\end{equation}

Though this result can describe the goodness of fit very well and guide us to choose the sampling numbers, we need to give two artificial criterions, the deviation value $d$ and the probability $ P\{D_n>d\}$, in order to determine $n$. In addition, this method do not take the message importance of samples into account, which makes the process of information collection not intuitionistic.

In this paper, we consider this problem from the perspective of DMIM. Firstly, we define
\begin{equation}\label{equ:XNX define}
\gamma \left( n \right)  ={{l\left( {\sum\limits_{i = 1}^n {{X_i}} } \right)} / {l(X)}}  .
\end{equation}
as relative importance of these $n$ sample points. According to central-limit theorem \cite{ross2014first}, when $n$ is big enough, $\sum\nolimits_{i =1}^n {{X_i}} $ approximately obeys normal distribution $N(n\mu,n\sigma^2)$. In fact, when $\sqrt{n}\sigma$ is not too small (such a condition is satisfied because $n$ is big enough), $l\left( {\sum\nolimits_{i = 1}^n {{X_i}} } \right) \approx {e^{ - {1 \over {2\sqrt {\pi n} \sigma }}}}$ according to (\ref{equ:normal jieduan 2}). Hence
\begin{equation}\label{equ:XNX define 1}
\gamma(n)={{{e^{ - {1 \over {2\sqrt {\pi n} \sigma }}}}} \over {l\left( X \right)}}.
\end{equation}

We find $\gamma(n)$ increases rapidly firstly, and then increases slowly by analyzing its monotonicity. Moreover, we obtain
\begin{equation}\label{equ:XNX infty}
\gamma(\infty)=\mathop {\lim }\limits_{n \to \infty } \gamma \left( n \right) = \mathop {\lim }\limits_{n \to \infty } {{{e^{ - {1 \over {2\sqrt {\pi n} \sigma }}}}} \over {l\left( X \right)}} = {1 \over {l\left( X \right)}},
\end{equation}
which means $\gamma(n)$ reaches limit as $n \to \infty$. In fact, these two points are consistent with the characteristic of data fitting. Both $\gamma(n)$ and data fitting have the law of diminishing of marginal utility. Furthermore, the goodness of fit can not increase unboundedly and it reaches the upper bound when the number of sampling points approaches infinity. DMIM is bounded, while Shannon entropy and R{\'e}nyi entropy do not possess these characteristic. In conclusion, we adopt $\left| \gamma(\infty)-\gamma(n)\right|$ to describe the goodness of fit.


\begin{thm}\label{thm:low bound number}
  $X_1,X_2,X_3,\dots,X_n$ are the $n$ sampling of a continuous random variable $X$, whose density is $f(x)$. If $\left| \gamma(\infty)-\gamma(n)\right| \leq \varepsilon$, we will obtain
  \begin{equation}
   P\left\{D_n>\sqrt{2\pi\sigma^2 \ln{\frac{19}{9\beta}}} \ln{\frac{1}{1-\varepsilon}}\right\} < \beta.
  \end{equation}
\end{thm}
\begin{proof}
In fact, a upper bound of $ P\{D_n>d\}$ is given by
\begin{flalign}\label{equ: up_bound sum}
  & P\left\{ {{D_n} > d} \right\} \approx 2\sum\limits_{k = 1}^\infty  {{{\left( { - 1} \right)}^{k - 1}}{e^{ - 2n{k^2}{d^2}}}}   \\
  &  = 2\sum\limits_{m = 1}^\infty  {\left( {{e^{ - 2n{{\left( {2m - 1} \right)}^2}{d^2}}} - {e^{ - 2n{{(2m - 1 + 1)}^2}{d^2}}}} \right)}   \tag{\theequation a}\label{equ: up_bound sum a}\\
  &  < 2\sum\limits_{m = 1}^\infty  {{e^{ - 2n{{\left( {2m - 1} \right)}^2}{d^2}}}}   \tag{\theequation b}\label{equ: up_bound sum c}\\
  &  \le 2\sum\limits_{m = 1}^\infty  {{e^{ - 4n{d^2}\left( {2m - 1} \right) + 2n{d^2}}}}   \tag{\theequation c}\label{equ: up_bound sum d}\\
  &  = 2\sum\limits_{m = 1}^\infty  {{e^{ - 8n{d^2}m + 6n{d^2}}}}    = {{2{e^{ - 2n{d^2}}}} \over {1 - {e^{ - 8n{d^2}}}}}  .\tag{\theequation d}\label{equ: up_bound sum f}
\end{flalign}
(\ref{equ: up_bound sum c}) is obtained for the fact that ${e^{ - 2n{{(2m - 1 + 1)}^2}{d^2}}} > 0$. (\ref{equ: up_bound sum d}) requires $ - 2n{d^2}{\left( {2m - 1} \right)^2} \le  - 4n{d^2}\left( {2m - 1} \right) + 2n{d^2}$. Such a condition is satisfied because $ - 2n{d^2}{\left( {2m - 1 - 1} \right)^2} \le 0$.

This means, we only need to check ${{{e^{ - 2n{d^2}}}} \over {1 - {e^{ - 8n{d^2}}}}} \le \beta $ holds. 

Substituting (\ref{equ:XNX define 1}) and (\ref{equ:XNX infty}) in $\left| \gamma(\infty)-\gamma(n)\right| \leq \varepsilon$, we get
\begin{equation}\label{equ:NMIM_number_1}
   \left| {{1 \over {l\left( X \right)}} - {{{e^{ - {1 \over {2\sqrt {\pi n} \sigma }}}}} \over {l\left( X \right)}}} \right| \le \varepsilon \Rightarrow n \ge {1 \over {4\pi {\sigma ^2}{{\ln }^2}\left( {1 - \varepsilon l\left( X \right)} \right)}}.
\end{equation}

Because $0\leq l(X) \leq 1$, we obtain
\begin{equation}\label{equ:number_choose}
n \ge {1 \over {4\pi {\sigma ^2}{{\ln }^2}\left( {1 - \varepsilon l\left( X \right)} \right)}} \ge {1 \over {4\pi {\sigma ^2}{{\ln }^2}\left( {1 - \varepsilon } \right)}}.
\end{equation}
Letting
\begin{equation}\label{equ:NMIM_proof d}
d=\sqrt{2\pi\sigma^2 \ln{\frac{19}{9\beta}}} \ln{\frac{1}{1-\varepsilon}},
\end{equation}
we have
 \begin{flalign}\label{equ:NMIM_proof 1}
   2n{d^2} &\ge 2{{2\pi {\sigma ^2}\ln {{19} \over {9\beta }}{{\ln }^2}(1 - \varepsilon )} \over {4\pi {\sigma ^2}{{\ln }^2}(1 - \varepsilon )}}   \Rightarrow  {e^{ - 2n{d^2}}} \le {{9\beta } \over {19}}.
\end{flalign}
It is easy to check
\begin{equation}\label{equ:NMIM_proof 2}
\beta {\left( {{e^{ - 2n{d^2}}}} \right)^4} +2 {e^{ - 2n{d^2}}} - \beta  \le 0,
\end{equation}
when $\beta \le {{19} \over 9}\root 4 \of {{1 \over {19}}} \approx1.0112$. In fact, $\beta$ is a probability value, so we usually take $\beta \le 1$. Therefore, (\ref{equ:NMIM_proof 2}) holds all the time.

Hence,
\begin{equation}
{2{{e^{ - 2n{d^2}}}} \over {1 - {e^{ - 8n{d^2}}}}} \le \beta .
\end{equation}
Based on the discussions above, we get
  \begin{equation}
   P\left\{D_n>\sqrt{2\pi\sigma^2 \ln{\frac{19}{9\beta}}} \ln{\frac{1}{1-\varepsilon}}\right\}< \beta.
  \end{equation}
\end{proof}

\begin{rem}
Due to (\ref{equ:NMIM_proof d}), we obtain
 \begin{flalign}\label{equ:NMIM_proof relation}
   \varepsilon  &= 1 - {e^{-d{{\left( {2\pi {\sigma ^2}\ln {{19} \over {9\beta }}} \right)}^{ - 1/2}}}},   \\
   \beta  &= {{19} \over 9}{e^{ - {{{d^2}} \over {2\pi {\sigma ^2}{{\ln }^2}(1 - \varepsilon )}}}}   \tag{\theequation a}\label{equ:NMIM_proof relation a}.
\end{flalign}
Therefore, there is a ternary relation among $d$, $\beta$ and $\varepsilon$. If two of them are known, the third one can be obtained.
\end{rem}

\begin{rem}
For arbitrary positive number $d$ and $\beta \le1$, one can always find a $\varepsilon_0$, which can be obtained by (\ref{equ:NMIM_proof relation}), when $\varepsilon \le \varepsilon_0$, $P\left\{ {{D_n} > d} \right\} < \beta $ holds.
\end{rem}

\begin{rem}
When $\varepsilon$ tends zero, which means $n\to \infty$, at this time, $P\left\{D_n>0\right\}=0$. Therefore, the real distribution is equal to empirical distribution with probability $1$ as $\varepsilon \to 0$. That is,
\begin{equation}
\hat F_n(x) \to F(x) \quad as \quad \varepsilon \to 0.
\end{equation}
\end{rem}

Actually, the DMIM deviation characterizes the process of collection information. With the growth of sampling number, the information gathers, and the empirical distribution approaches real distribution at the same time. In particular, when $n \to \infty$, all the information about the real distribution will be obtained. In this case, the empirical distribution is equal to real distribution, naturely.

\begin{rem}
For arbitrary continuous random variable with variance $\sigma^2$, if the max allowed DMIM deviation is $\varepsilon$, the sampling number must be bigger than $1/(4 \pi \sigma^2 \ln^2(1-\varepsilon) $ according to (\ref{equ:number_choose}). 
\end{rem}
The sampling number only depends on one artificial criterion, the DMIM deviation, and the variance is the own attribute of the distribution. Furthermore, the sampling number has nothing to do with the distribution form.

\section{Numerical Results} \label{sec:five}
In this section, we present some numerical results to validate the above results in this paper.

\subsection{The properties of DMIM}
 Fig. \ref{fig:normal_big_sigma_err2} shows relative error for approximation $e^{-1/(2\sqrt{\pi}\sigma)}$ and $1-1/(2\sqrt{\pi}\sigma)$ when $\sigma$ increases from $0.1$ to $10$. When $\sigma$ is not too small ( $\sigma>1$ for $e^{-1/(2\sqrt{\pi}\sigma)}$ and $\sigma>2$ for $1-1/(2\sqrt{\pi}\sigma)$), the relative error of both approximations is smaller than $1\%$. The relative error decreases with increasing of $\sigma$ for these two approximate values. When $\sigma<6.5$, the error of $e^{-1/(2\sqrt{\pi}\sigma)}$ is smaller than that of $1-1/(2\sqrt{\pi}\sigma)$ and the opposite is true when $\sigma>6.5$. In summary, $e^{-1/(2\sqrt{\pi}\sigma)}$ is a good approximation when $\sigma$ is not too small and $1-1/(2\sqrt{\pi}\sigma)$ is an excellent approximation when $\sigma$ is big enough.

\begin{figure}
  \centerline{\includegraphics[width=8.0cm]{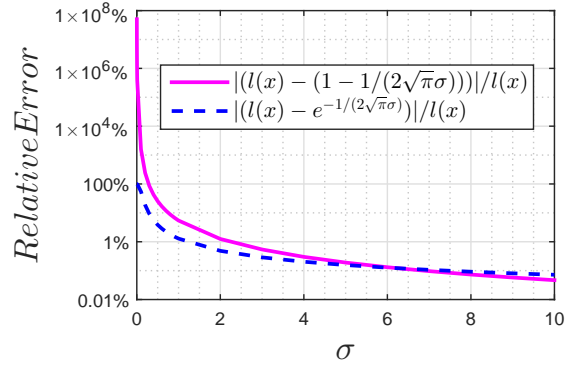}}
  \caption{Relative error vs. $\sigma$.}\label{fig:normal_big_sigma_err2}
\end{figure}

Fig. \ref{fig:smae_var} shows the DMIM of uniform distribution, Gauss distribution and negative exponential distribution when the variance increases from $0.1$ to $100$. It can be observed that the DMIM is subject to the variance. The DMIM increases with the increasing of variance for these three distributions. Among them, the DMIM in Gauss distribution is the largest and that in negative exponential distribution is the smallest.

\begin{figure}
  \centerline{\includegraphics[width=8.0cm]{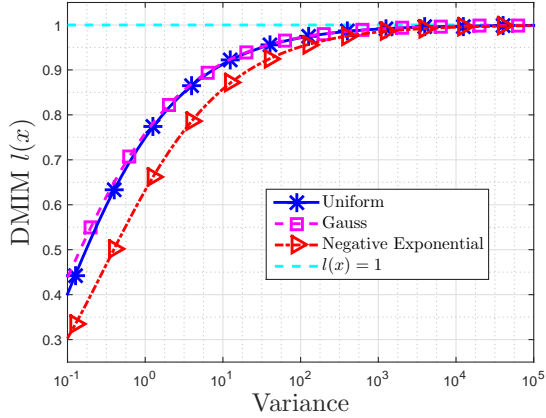}}
  \caption{DMIM $l(X)$ vs. Variance.}\label{fig:smae_var}
\end{figure}

\subsection{Goodness of fit by DMIM}
Next we focus on conducting Monte Carlo simulation by computer to validate our results about goodness of fit. The samples are independent identically distributed, each having variance $\sigma^2$. $\sigma$ is $1$ or $2$. $\lambda=1/\sigma$ in negative exponential distribution, and the density of uniform distribution is $1/(2\sqrt{3}\sigma)$. The mean of normal distribution and uniform distribution is zero. The DMIM deviation $\varepsilon$ is varying from $0.001$ to $0.1$. For each value of $\varepsilon$, the simulation is repeated $10000$ times.

Fig. \ref{fig:Number_normal} shows the relationship between the probability of error bound $P\{D>d\}$ and DMIM deviation $\varepsilon$. Some observations can be obtained. The result that the goodness of fit is controlled by the DMIM deviation is true. That is, the probability of error bound decreases with the decreasing of DMIM deviation. In fact, this process can be divided into three phases. In phase one, in which $\varepsilon$ is very small (e.g. $\varepsilon<10^{-2.8}$ when $d=0.01$ and $\sigma=1$), $P\{D>d\}$ is close to zero. In phase two, $\varepsilon$ is neither too small nor too large (e.g. $10^{-2.8}<\varepsilon<10^{-2}$ when $d=0.01$ and $\sigma=1$). In this case, $P\{D>d\}$ increases rapidly from zero to one. In the third phase, in which $\varepsilon$ is large (e.g. $\varepsilon>10^{-2}$ when $d=0.01$ and $\sigma=1$), $P\{D>d\}$ approaches one. When $d=0.01$ and $\sigma=1$, $P\{D>d\}$ in these three distributions is very close to each other, whose upper bound is indeed $\beta$ (obtained by (\ref{equ:NMIM_proof relation a})). For the same distribution, if the standard deviation is the same, $P\{D>d\}$ will decrease with increasing of $d$ when $P\{D>d\}<1$. Furthermore, for the same $d$, the probability of error bound increases with increasing of the given standard deviation.

\begin{figure}
  \centerline{\includegraphics[width=8.0cm]{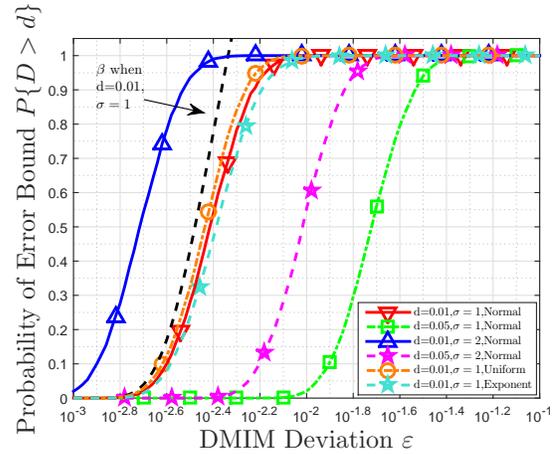}}
  \caption{Probability of error bound $P\{D>d\}$ vs. DMIM deviation $\varepsilon$.}\label{fig:Number_normal}
\end{figure}

\section{Conclusion} \label{sec:six}
In this paper, we discussed the distribution goodness-of-fit in the view of information collection, where the message importance is taken into account. Similar to differential entropy, DMIM was proposed as an measure of message importance for continuous random variable to help us describe the information flows during sampling. Then, uniform, normal and negative exponential distribution were discussed as typical examples, and high-precision approximate values for DMIM of normal distribution were given. Finally, we proved that the divergence between the empirical distribution and the real distribution is controlled by the DMIM deviation. Compared with Kolmogorov-Smirnov test, the new method based on DMIM gives us another viewpoint of information collection because it visually shows the information flow with the increasing of sampling points, which helps us to design sampling strategy for the actual system of big data.


\end{document}